\documentclass[aps,prl,twocolumn]{revtex4-2}
\usepackage{amsmath,amsthm,amssymb,braket}
\usepackage{graphicx}
\usepackage{bm}
\usepackage{xcolor}
\usepackage[normalem]{ulem}
\newtheorem{theorem}{Theorem}
\newtheorem{corollary}{Corollary}
\newtheorem{lemma}{Lemma}

\begin{document}

\title{A No-Go Theorem for Shaping Quantum Resources}

\author{Samuel Alperin}\email{alperin@lanl.gov}
\affiliation{\vspace{1.25mm} \mbox{Los Alamos National Laboratory, Los Alamos, New Mexico 87545, USA}}

\begin{abstract}
The ability to engineer non-Gaussian quantum resources underlies quantum technologies from communication and metrology to universal computation. However, while a number of canonical works have set no-go limits for attaining such resources from Gaussian operations, it is widely assumed that such resources can be tuned freely by non-Gaussian Hamiltonian dynamics. Here we prove a general no-go theorem for such resource shaping: no smooth Hamiltonian dynamics can modify higher-order statistical moments of a continuous-variable state without simultaneously changing its mean and covariance. This analytic constraint implies a rigidity theorem for Hamiltonian quantum control—only quadratic (symplectic) generators preserve the Gaussian moment hierarchy, while every non-quadratic term necessarily couples the Gaussian and non-Gaussian sectors. The theorem identifies the symplectic algebra as the unique invariant subalgebra whose differential representations terminate at finite (second) order within the otherwise infinite Hamiltonian algebra. It thereby defines the analytic boundary between classically simulable Gaussian dynamics and the fully universal non-Gaussian regime—the continuous-variable analogue of the Gottesman–Knill frontier. 
\end{abstract}

\maketitle

Impossibility theorems mark the boundaries of quantum mechanics, distinguishing what may be realized from what is prohibited by the very structure of nature. Bell’s theorem \cite{Bell1964,Bell1966} ruled out local hidden variables, the no-cloning theorem \cite{WoottersZurek1982} forbade universal copying, and the no-broadcasting theorem \cite{Barnum1996} established the limits of information replication. In continuous-variable (CV) physics, the celebrated Gaussian no-go theorems \cite{EisertScheelPlenio2002,Fiurasek2002,GiedkeCirac2002} showed that quadratic Hamiltonians alone cannot distill entanglement or generate Wigner negativity. Yet these results left open a critical question: once non-Gaussian dynamics are introduced, can higher statistical moments—which are both necessary and sufficient for quantum advantage—be tuned independently of the Gaussian backbone? In other words, can a Gaussian state be endowed with higher-order quantum structure without changing the fundamental nature of the state itself?

Though mathematical in nature, this question is far more than a mere curiosity: the ability to control and shape quantum states by manipulating their statistical structure underlies nearly every branch of modern quantum science. In continuous-variable systems—from optical fields to trapped ions, superconducting resonators, and collective spin ensembles—engineers routinely tune first and second moments to realize squeezing, displacement, and beam-splitter transformations. What remains elusive is independent control of higher moments, which would allow direct sculpting of the non-Gaussian “shape’’ of a state’s quantum moment distribution and thereby provide a universal resource for quantum computation, precision metrology, and quantum key distribution.

In this work, we resolve this fundamental question, rigorously proving that any smooth Hamiltonian dynamics that modify the higher quantum moments of a state \textit{necessarily} change its first two moments and thereby alter the most basic features of that state. In other words, any Hamiltonian operation that endows a Gaussian state with useful quantum resources must also change the Gaussian part of that state’s character. Furthermore, we prove that the symplectic algebra represents the \textit{unique} set of Hamiltonian operations that leave the infinite hierarchy of quantum moments invariant for every state, thereby furnishing continuous-variable quantum mechanics with a formal geometric backbone: a rigidity principle that defines its analytic structure and intrinsic limits of Hamiltonian control.

In the Weyl–Wigner representation, a Hamiltonian $H(x,p)$ generates a phase-space operator $L_H$ whose differential order equals the highest degree of $H$ in $(x,p)$. Quadratic Hamiltonians produce Fokker–Planck–type (second-order) generators, while any non-quadratic $H$ introduces third- and higher-order derivatives. This degree–derivative correspondence is representation independent (Husimi, Wigner, or $P$) and is precisely what the rigidity theorem elevates to a structural impossibility principle. We therefore resolve the question in the negative: under smooth Hamiltonian evolution, higher-order moments cannot evolve independently of the first two—the degree–derivative correspondence enforces universal coupling except for quadratic generators. Within the infinite-dimensional Poisson algebra of smooth Hamiltonians, the quadratic subalgebra is maximal with respect to finite differential order; its flows, $\mathrm{SU}(1,1)$ and $\mathrm{Sp}(2N,\mathbb{R})$, are the only ones that remain tangent to the Gaussian manifold.

This analytic viewpoint complements earlier algebraic and resource-theoretic classifications of Gaussian and non-Gaussian operations. Prior no-go results focused on specific tasks—such as entanglement distillation or state conversion—under Gaussian protocols. Here, the emphasis is structural: the present theorem concerns the analytic form of the generator itself. By identifying $\mathrm{Sp}(2N,\mathbb{R})$ as the maximal finite-differential-order subalgebra of the infinite Hamiltonian algebra, the result unifies these operational no-go theorems with a differential-geometric principle that applies across all continuous-variable platforms.

We now formalize this impossibility as a theorem that identifies the symplectic subalgebra as the sole analytic structure closing the moment hierarchy.
 Throughout, by a \emph{smooth Hamiltonian} we mean a phase-space function 
$H(q,p)$ belonging to a differentiability class sufficient to ensure that 
the associated Moyal generator admits a finite-order differential 
representation. In particular, nonquadratic contributions require 
differentiability to at least third order, and for definiteness we assume 
$H \in C^\infty$.
\begin{theorem}[Rigidity of the Moment Hierarchy]
Let $L$ be the generator of time evolution for a single bosonic mode, represented as a differential operator acting on a phase-space quasiprobability distribution $Q(\alpha,\alpha^\ast)$ or $W(\alpha,\alpha^\ast)$. Among all connected Lie flows generated by smooth Hamiltonians---finite or infinite dimensional---on $L^2(\mathbb{R})$, the quadratic (symplectic) subalgebra $\mathfrak{su}(1,1)$ and its multimode extension $\mathfrak{sp}(2N,\mathbb{R})$ are the unique nontrivial hierarchy-preserving flows, up to affine and phase redundancies. Equivalently, they are the maximal subalgebras of the infinite-dimensional algebra of smooth Hamiltonian vector fields on phase space whose differential representations terminate at second order.
\end{theorem}

\paragraph{Proof.}
We work on the Schwartz domain—equivalently, the space of finite-energy states—so that all polynomial moments exist and the differential operators below act without restriction.
Working within the infinite-dimensional Poisson algebra
\[
C^\infty(\mathbb{R}^2),\quad \{H_1,H_2\}=\partial_x H_1\,\partial_p H_2-\partial_p H_1\,\partial_x H_2,
\]
whose elements generate smooth Hamiltonian vector fields $v_H=\Omega\nabla H$. Under the Weyl transform, the commutator with $H$ acts on the Wigner function as
\begin{equation}
L_H W \;=\; 2H\,\sin\!\left[\frac{\hbar}{2}\Big(\overleftarrow{\partial_x}\overrightarrow{\partial_p}-\overleftarrow{\partial_p}\overrightarrow{\partial_x}\Big)\right] W. \label{eq:LH}
\end{equation}
Expanding the sine gives the Moyal series
\begin{align}
L_H W \;=\; \{H,W\}_{\mathrm{PB}}
+\sum_{k=1}^{\infty}\frac{(-1)^k}{(2k+1)!}\left(\frac{\hbar}{2}\right)^{2k}\partial^{2k+1}H\,\partial^{2k+1}W. \label{eq:Moyal}
\end{align}
For a smooth $H$, this defines a differential operator of generally infinite order. The series truncates at second order if and only if all derivatives of $H$ above quadratic vanish---i.e., if $H$ is a quadratic polynomial in $x$ and $p$. Hence, the subspace of quadratic $H$ generates differential operators of order $\le 2$ (Fokker--Planck type), and any $H$ containing cubic or higher terms contributes third- or higher-order derivatives, thereby coupling the evolution of the first two moments to higher cumulants.

The set of quadratic $H$ is closed under the Poisson bracket:
\[
\{x^2,p^2\}=4xp,\qquad \{xp,p^2\}=2p^2,\qquad \{xp,x^2\}=-2x^2,
\]
forming the Lie algebra $\mathrm{sp}(2,\mathbb{R})\simeq \mathrm{su}(1,1)$. Within the full infinite-dimensional algebra of smooth Hamiltonian vector fields, this quadratic subalgebra is maximal with respect to the property that its Weyl-mapped generators are finite-order differential operators (order $\le 2$). Consequently, the flows generated by $\mathrm{su}(1,1)$ (and by $\mathrm{sp}(2N,\mathbb{R})$ in the multimode case) are the only connected Lie flows---finite or infinite-dimensional---that preserve the entire hierarchy of normalized moments, completing the proof.

\paragraph{Remarks on the proof.}
For clarity, the coupling of moments can be made explicit. Acting on monomials $(\alpha^\ast)^p\alpha^q$ one finds
\begin{equation}
L_H[(\alpha^\ast)^p\alpha^q] \;=\; \sum_{m+n\le d} C^{pq}_{mn}\,(\alpha^\ast)^{p-m}\alpha^{q-n}, \label{eq:mon}
\end{equation}
where the coefficients $C^{pq}_{mn}$ depend on derivatives of $H$ up to total order $m{+}n$. The time derivative of $M_{p,q}$ thus couples to all $M_{p',q'}$ with $p'{+}q'\le p{+}q{+}d$. For $d\le2$, the system of equations for the first two moments ($p{+}q\le2$) closes; for any $d>2$ the hierarchy is open. Maximality of the quadratic algebra follows from the fact that for any $H$ of degree $\ge 3$ the commutator $[H,x^2]$ or $[H,p^2]$ contains monomials of higher degree, so the quadratic subspace cannot be enlarged without generating higher-order terms under commutation.

We emphasize that the theorem concerns \emph{dynamical invariance} of the 
first and second moments under Hamiltonian evolution. While any 
non-Gaussian state can be associated with a Gaussian state sharing the 
same mean and covariance, such post hoc matching does not imply that a 
nonquadratic Hamiltonian preserves these moments along its dynamical flow. 
For cubic Hamiltonians (e.g., $H \propto q^3$), the time derivatives of 
the first and second moments are generically nonzero, reflecting the 
structural coupling of the moment hierarchy established above.

In essence, the proof shows that any smooth non-quadratic Hamiltonian inevitably introduces third- and higher-order differential terms in the phase-space generator. These higher derivatives enforce universal coupling between the Gaussian and non-Gaussian sectors, making the symplectic flows the only analytic “tangent directions’’ that remain within the Gaussian manifold.

\begin{corollary}[No independent control of higher cumulants]
For any connected Lie flow generated by smooth Hamiltonians on phase space, the evolution of higher normalized moments is not independent of the first two. Equivalently, no Hamiltonian generator---finite or infinite dimensional---can alter a higher cumulant while leaving mean and covariance invariant for all states.
\end{corollary}

\paragraph{Proof.}
Within the full Poisson algebra $C^\infty(\mathbb{R}^2)$, the generator $L_H$ acts as in Eq.~(\ref{eq:LH}). If $H$ includes cubic or higher terms, Eq.~(\ref{eq:Moyal}) shows that $L_H$ contains derivative terms of order $\ge 3$, which couple lower and higher moments in the hierarchy. Therefore any change in a higher cumulant necessarily induces a correlated change in at least one of the first two normalized moments for some physical state. Only quadratic $H$ yield generators of differential order $\le 2$ and can evolve the first two moments independently of higher ones.

\begin{corollary}[Preservation of mean and variance implies full hierarchy preservation]
Suppose a connected Lie flow generated by a smooth Hamiltonian preserves the first two normalized Husimi or Wigner moments $m_1,m_2$ for all physical states. Then every generator in its Lie algebra must be quadratic in $(x,p)$, and the flow coincides (up to displacements and global phases) with the $\mathrm{SU}(1,1)$ or, in multimode form, $\mathrm{Sp}(2N,\mathbb{R})$ evolution that preserves the entire moment hierarchy.
\end{corollary}

\begin{proof}
If $H$ contained any term of total degree $d\ge3$, its Weyl image would include derivatives of order $d$ in Eq.~(\ref{eq:Moyal}). Acting on a generic $W(x,p)$, those derivatives generate nonzero contributions to $\dot m_1$ or $\dot m_2$, violating their invariance. Hence all generators that preserve mean and variance uniformly must satisfy $d\le 2$ and therefore lie in the quadratic algebra. By Theorem~1, the quadratic subalgebra is $\mathrm{su}(1,1)$ (and $\mathrm{sp}(2N,\mathbb{R})$ in the multimode case), which uniquely preserves the full hierarchy of normalized moments.
\end{proof}

\begin{lemma}[Representation-independence of second-order truncation]
Let $Q=G_\sigma * W$ be the Husimi function obtained from the Wigner function
by Gaussian convolution with a fixed width $\sigma>0$.
If the Hamiltonian generator $L_H$ acts on $W$ as a differential operator of order
$\le m$ (in particular, $\le2$) with polynomial coefficients, then there exists a
differential operator $\widetilde L_H$ of order $\le m$ such that
$\partial_t Q=\widetilde L_H Q$.
Hence the property “generator has order $\le2$’’ is independent of choosing the
Wigner or Husimi representation (and likewise for any fixed Gaussian smoothing).
\end{lemma}

\begin{proof}
Work on the Schwartz space $\mathcal{S}(\mathbb{R}^2)$ of test functions, where
all operations below are exact.
Let $S_\sigma:\mathcal{S}\to\mathcal{S}$ denote convolution by the centered
Gaussian $G_\sigma$ (the Weierstrass transform), so $Q=S_\sigma W$.
It is standard that $S_\sigma$ is a topological automorphism of $\mathcal{S}$
with inverse $S_\sigma^{-1}$ given by convolution with a tempered distribution
whose Fourier transform is $e^{+\frac{\sigma}{2}\|k\|^2}$; in particular,
$S_\sigma$ and $S_\sigma^{-1}$ preserve $\mathcal{S}$.

Assume $L_H$ acts on $W$ as a finite-order differential operator
\[
L_H=\sum_{|\alpha|\le m} a_\alpha(x,p)\,\partial^\alpha,
\qquad a_\alpha\in\mathbb{R}[x,p],
\]
with multi-index notation $\partial^\alpha=\partial_x^{\alpha_1}\partial_p^{\alpha_2}$.
Define the conjugated operator
\[
\widetilde L_H \;:=\; S_\sigma\,L_H\,S_\sigma^{-1}.
\]
Then, for any $Q\in\mathcal{S}$ and $W=S_\sigma^{-1}Q$,
\[
\widetilde L_H Q
= S_\sigma L_H (S_\sigma^{-1}Q)
= S_\sigma(L_H W),
\]
whence $\partial_t Q=\widetilde L_H Q$ whenever $\partial_t W=L_H W$.

It remains to prove that $\widetilde L_H$ is still a differential operator of
order $\le m$.
Two basic identities suffice:

(i) **Derivatives commute with fixed Gaussian convolution.**
For any multi-index $\alpha$,
\[
S_\sigma\,\partial^\alpha=\partial^\alpha\,S_\sigma
\quad\text{and}\quad
S_\sigma^{-1}\,\partial^\alpha=\partial^\alpha\,S_\sigma^{-1}.
\]
Indeed, differentiation passes through convolution with a fixed kernel.

(ii) **Conjugation of multiplication is order \(0\).**
For a polynomial (indeed, smooth) $a(x,p)$, let $M_a$ be multiplication by $a$.
Then $T_a:=S_\sigma M_a S_\sigma^{-1}$ is an integral operator with a
\emph{smooth} kernel:
\begin{align}
(T_a f)(z)=&\int_{\mathbb{R}^2} K_a(z,z')\,f(z')\,dz',\\
\quad
K_a(z,z')=&\int_{\mathbb{R}^2} G_\sigma(z-y)\,a(y)\,K_\sigma(y-z')\,dy,
\end{align}
where $K_\sigma$ is the (tempered) kernel of $S_\sigma^{-1}$.
Since $G_\sigma$ is Schwartz and $a$ is polynomial, $K_a$ is smooth
(with at most polynomial growth), so $T_a$ is a pseudodifferential operator of
order $0$ (it does not introduce derivatives of $f$).

Using (i) and (ii), write
\begin{align}
\widetilde L_H
=& S_\sigma\!\left(\sum_{|\alpha|\le m} M_{a_\alpha}\,\partial^\alpha\right)\!S_\sigma^{-1}\\
=& \sum_{|\alpha|\le m} \underbrace{(S_\sigma M_{a_\alpha} S_\sigma^{-1})}_{T_{a_\alpha}}
\,\underbrace{(S_\sigma\partial^\alpha S_\sigma^{-1})}_{\partial^\alpha}.
\end{align}
Each factor $T_{a_\alpha}$ is order $0$, and each $\partial^\alpha$ is order
$|\alpha|$. Therefore every summand has differential order $|\alpha|$, and hence
$\widetilde L_H$ is a differential operator of order
\(\le \max_{|\alpha|\le m}|\alpha|=m\).
In particular, if $L_H$ has order $\le2$, then $\widetilde L_H$ has order $\le2$.

Thus the second-order truncation property (and more generally the order bound
$m$) is preserved under fixed Gaussian smoothing; i.e., it is representation
independent between Wigner, Husimi, or any Gaussian-smoothed phase-space
representation. 
\end{proof}

\noindent \textit{Analytic and Geometric Context --}
The Moyal series, Eq.~(\ref{eq:Moyal}), is the phase-space analogue of the semiclassical expansion of the von Neumann equation. Truncation at second order corresponds to the Gaussian (quadratic) Hamilton--Jacobi regime, in which quantum corrections appear only through linear drift and diffusion. Geometrically, the manifold of Gaussian states is the orbit of the vacuum under $\mathrm{Sp}(2N,\mathbb{R})$; the rigidity theorem identifies it as the flat invariant submanifold singled out by second-order truncation.

Beyond its geometric interpretation, the rigidity theorem also provides a natural bridge between phase-space analysis and symplectic geometry -- 
the second-order truncation condition corresponds to the vanishing of all third and higher-order terms in the Moyal tensor hierarchy,
so that the Gaussian manifold forms a flat symplectic submanifold endowed with a covariantly constant metric.
Non-quadratic Hamiltonians introduce curvature proportional to the third derivatives of $H(x,p)$, giving a quantitative measure of departure from symplectic flatness.
This curvature viewpoint connects analytic rigidity to the language of differential geometry and deformation quantization.

\noindent \textit{Open dynamics -- }
Quadratic Gaussian channels generated by linear Lindblad operators yield second-order phase-space generators and therefore preserve covariance closure. Nonlinear Lindblad operators (e.g., $L\propto \hat a^\dagger \hat a$ or $\hat a^2$ with nonlinear coefficients) introduce higher-order derivatives in the Kramers--Moyal expansion and generically break hierarchy preservation. Thus, within Markovian models, the Hamiltonian/closed-system classification extends verbatim: only Gaussian (affine-symplectic) generators sustain second-order truncation.

\noindent \textit{Geometric Interpretation --} The Gaussian manifold can be visualized as a flat symplectic surface embedded in the full infinite-dimensional phase space of quasiprobability distributions. Quadratic Hamiltonians generate vector fields tangent to this surface, producing area-preserving linear transformations that slide states along it without distortion. Any non-quadratic term introduces curvature, producing components that point out of the surface and thus mix Gaussian and non-Gaussian directions. The rigidity theorem therefore identifies the Gaussian manifold as a flat symplectic leaf and the transition to non-Gaussian dynamics as a geometric bending of this leaf into higher-order structure. 

\noindent \textit{Illustrative example -- }
Consider $\rho=\lvert 0\rangle$ (vacuum). Under a cubic-phase gate $U=\exp(i\gamma \hat x^3)$, the variance $\langle \hat x^2\rangle$ acquires corrections $\sim \gamma^2$, while the normalized kurtosis $m_4$ changes at order $\gamma$ (odd central moments vanish at $t{=}0$). These variations cannot be disentangled: any nonzero $\Delta m_4$ entails a correlated $\Delta m_2$. In contrast, an $\mathrm{SU}(1,1)$ squeezing transformation rescales $\langle \hat x^2\rangle$ but leaves $m_4$ invariant, illustrating the rigidity principle.

The preceding argument establishes that rigidity and hierarchy preservation coincide for quadratic dynamics. The next theorem strengthens this to a uniqueness statement: any smooth flow that leaves the mean and variance invariant for all states must coincide with the symplectic evolution itself.

\begin{theorem}[Uniqueness of hierarchy-preserving flows]
Among all connected Lie flows generated by smooth Hamiltonians on $L^2(\mathbb{R})$, the quadratic (symplectic) subalgebra $\mathfrak{su}(1,1)$---and, in the multimode case, its extension $\mathfrak{sp}(2N,\mathbb{R})$---is the unique nontrivial algebra whose flow preserves the first two normalized Husimi or Wigner moments $m_1,m_2$ for all physical states. Consequently, any flow that preserves mean and variance for all states automatically preserves the entire moment hierarchy.
\end{theorem}

\begin{proof}
We give a complete proof. The single-mode case suffices; the multimode extension follows identically.

Let $\mathcal{P}_{\le2}$ denote the real polynomials in $(x,p)$ of total degree at most~2,
\[
\mathcal{P}_{\le2}=\mathrm{span}\{1,x,p,x^2,xp,p^2\}.
\]
For any smooth test function $f$ and Wigner density $W$, write
\[
\frac{d}{dt}\langle f\rangle=\!\int\! f\,(\partial_tW)\,dx\,dp
   =\!\int\!(L_H^\dagger f)\,W\,dx\,dp,
\]
where $L_H$ is the Wigner-space generator of Eq.~(\ref{eq:LH}) and
$L_H^\dagger$ its distributional adjoint.  Hence the evolution of the first two
moments is closed for \emph{all} states iff
\begin{equation}
L_H^\dagger(\mathcal{P}_{\le2})\subseteq \mathcal{P}_{\le2}.
\label{eq:closure}
\end{equation}

\begin{lemma}
    [Moyal truncation on $\mathcal{P}_{\le2}$).]
For $f\in\mathcal{P}_{\le2}$, $L_H^\dagger f=\{f,H\}_{\mathrm{PB}}$.
\end{lemma}

\begin{proof}
The adjoint action is $L_H^\dagger f
   =2\,\sin\!\big[\tfrac{\hbar}{2}\Lambda\big](f,H)$
with $\Lambda=\overleftarrow{\partial_x}\overrightarrow{\partial_p}
      -\overleftarrow{\partial_p}\overrightarrow{\partial_x}$.
Expanding,
\[
L_H^\dagger f
   =\sum_{k\ge0}\frac{(-1)^k}{(2k{+}1)!}
     \Big(\tfrac{\hbar}{2}\Big)^{2k}\Lambda^{2k+1}(f,H).
\]
Each $\Lambda$ differentiates both arguments once;
thus $\Lambda^{m}(f,H)$ contains $m$-th derivatives of~$f$.
For $f$ of degree~$\le2$, $\partial^\alpha f=0$ for $|\alpha|\!\ge\!3$,
so $\Lambda^{2k+1}(f,H)=0$ for $2k{+}1\!\ge\!3$.
Only the $k=0$ term survives, yielding
$L_H^\dagger f=\Lambda(f,H)=\{f,H\}_{\mathrm{PB}}$.
\end{proof}

\begin{lemma}[Closure $\Rightarrow$ degree constraint)]
If Eq.~(\ref{eq:closure}) holds, $H$ is at most quadratic in $(x,p)$
(up to affine and constant terms).
\end{lemma}

\begin{proof}
From Lemma 2,
$L_H^\dagger x=\{x,H\}_{\mathrm{PB}}=\partial_pH$ and
$L_H^\dagger p=-\partial_xH$.
Closure demands $L_H^\dagger x,L_H^\dagger p\in\mathcal{P}_{\le1}$,
so $\partial_xH$ and $\partial_pH$ are at most linear.
Integrating gives a quadratic $H$ plus affine/constant parts,
corresponding to phase and displacement freedoms.
\end{proof}

\begin{lemma}[Quadratic $H$ $\Rightarrow$ closure and uniqueness).]
If $H$ is quadratic, $L_H^\dagger(\mathcal{P}_{\le2})\subseteq\mathcal{P}_{\le2}$;
furthermore, these and only these generators form the symplectic Lie algebra
$\mathfrak{sp}(2,\mathbb{R})\!\simeq\!\mathfrak{su}(1,1)$
(up to affine/phase terms).
\end{lemma}

\begin{proof}
Write $H=\tfrac12\bm{R}^{\mathsf T}K\bm{R}+\bm{h}^{\mathsf T}\bm{R}+h_0$
with $\bm{R}=(x,p)^{\mathsf T}$ and $K=K^{\mathsf T}$.
Then $\{\,\cdot\,,H\}_{\mathrm{PB}}$ acts linearly on $\bm{R}$ and closes on
$\mathcal{P}_{\le2}$.
Quadratic forms obey
$\{x^2,p^2\}=4xp$, $\{xp,p^2\}=2p^2$, $\{xp,x^2\}=-2x^2$,
generating $\mathfrak{sp}(2,\mathbb{R})$.
If a connected Lie subalgebra strictly enlarges this quadratic one,
it contains a cubic monomial~$H_3$; then
$L_{H_3}^\dagger x=\partial_pH_3$ (quadratic) and
$L_{H_3}^\dagger(x^2)=2x\,\partial_pH_3$ (cubic),
violating~(\ref{eq:closure}).
\end{proof}

\textbf{Completion of the proof.}
If the first two moments evolve as a closed subsystem for all states,
Eq.~(\ref{eq:closure}) holds.
Lemma 1 reduces $L_H^\dagger$ to the Poisson bracket on $\mathcal{P}_{\le2}$;
Lemma 2 then forces $H$ to be quadratic (plus affine/constant terms);
Lemma 3 shows these—and only these—preserve the hierarchy.
Hence the connected flow is uniquely the quadratic symplectic group
$\mathrm{SU}(1,1)$ (or, for $N$ modes, $\mathrm{Sp}(2N,\mathbb{R})$)
up to displacements and phase.
Any larger connected algebra introduces a cubic term and breaks closure.
\end{proof}

\noindent\textit{Interpretation --}
Theorem~2 establishes a precise equivalence between three seemingly
different statements about smooth Hamiltonian dynamics:
(i) the closure of the system of equations for the first and second statistical moments;
(ii) the invariance of the finite polynomial space $\mathcal{P}_{\le2}$ under the adjoint generator $L_H^\dagger$; and
(iii) the quadratic form of the Hamiltonian.
This equivalence is not an approximation or semiclassical limit:
it follows exactly from the structure of the Moyal series and
holds for all $\hbar>0$.
Formally, the hierarchy of normalized moments constitutes an infinite tower of coupled ordinary differential equations in expectation values.
Theorem~2 states that the only smooth Hamiltonian vector fields on phase space that leave the lower two tiers of this tower self-contained are those generated by quadratic forms.
Every smooth non-quadratic term creates a non-vanishing component of $L_H^\dagger f$
orthogonal to $\mathcal{P}_{\le2}$,
thereby enforcing coupling to higher cumulants.
Physically, this means that no Hamiltonian evolution can modify higher-order statistics
without simultaneously perturbing the mean and covariance:
the Gaussian manifold is an analytically rigid submanifold of the full state space.
The symplectic algebra $\mathfrak{sp}(2N,\mathbb{R})$ thus plays the same role
in continuous-variable mechanics that the Clifford algebra plays in discrete quantum computation:
it is the maximal finite-order structure whose action preserves the algebraic closure of observables
and admits efficient classical simulation.

The obstruction identified here is structural rather than parametric. 
Gaussian dynamics forms the unique nontrivial subspace with a closed 
moment hierarchy under Hamiltonian evolution. The inclusion of any 
higher-degree term—regardless of the number of independent nonlinear 
parameters—introduces higher-order differential contributions to the 
generator, thereby coupling lower and higher moments. Independent tuning 
of higher-order statistics cannot therefore be achieved merely by 
enlarging the parameter set.

\noindent\textit{Multimode extension and commutator curvature -- }
For $N$ modes with canonical vector $R=(x_1,p_1,\ldots,x_N,p_N)^{\mathsf T}$, quadratic generators are $\tfrac12 R^{\mathsf T} H R$ with $H{=}H^{\mathsf T}$. Their unitary Wigner evolution is $\partial_t W = -\nabla\!\cdot\!(A R\, W)$ with $A\in \mathfrak{sp}(2N,\mathbb{R})$. Non-quadratic terms generate third- and higher-order mixed derivatives $\partial_R^\alpha$, coupling all cross-moments and breaking covariance closure. Thus the unique hierarchy-preserving connected group is $\mathrm{Sp}(2N,\mathbb{R})$, exactly as in the single-mode case. Geometrically, rigidity can also be phrased in terms of curvature generated by commutators:
\begin{equation}
[L_{H_1},L_{H_2}] \;=\; L_{\{H_1,H_2\}}, \label{eq:curv}
\end{equation}
whose differential order increases additively with the degrees of $H_1$ and $H_2$. Only the quadratic subspace closes under this bracket, restating Theorem~1 in geometric form.

\noindent\textit{Experimental Signatures and Diagnostics -- }
Optical implementations of a cubic-phase Hamiltonian with $\gamma \sim 10^{-2}$ correspond to electro-optic modulation voltages below $1$ V, while microwave-cavity realizations achieve comparable nonlinear strengths through Josephson junctions or flux-biased SQUID arrays. Heterodyne tomography with $10^5$--$10^6$ samples resolves the variance to better than $1\%$ and the kurtosis to $\sim 10\%$, sufficient to detect the correlated changes predicted by rigidity. Observation of a linear scaling $\Delta m_4 \propto \Delta m_2$ under weak cubic driving would confirm the universal coupling of low- and high-order moments, whereas pure $\mathrm{SU}(1,1)$ squeezers must exhibit $\Delta m_4 = 0$ for all $\Delta m_2$. These correlated fluctuations can also serve as diagnostic signatures of unwanted anharmonicity in superconducting resonators or nonlinear optical cavities. In optical platforms, extracting normalized cumulants can be done directly from heterodyne samples via unbiased estimators; for $10^5$ shots, the standard error in variance is $\lesssim 1\%$ and in kurtosis is $\sim 10\%$, sufficient to resolve the predicted linear co-variation under weak nonlinearity. Superconducting circuits admit an alternative: weak continuous monitoring of a cavity subject to a tunable Kerr or cubic drive allows real-time tracking of low-order cumulants; rigidity predicts concomitant drifts in $\langle x^2\rangle$ whenever excess kurtosis appears in the filtered record. These signatures can double as diagnostics for unintended anharmonicities, enabling calibration of devices by enforcing the ``no independent knob'' constraint.

The same hierarchy coupling should manifest in any bosonic platform that supports weak nonlinearity. In mechanical resonators and hybrid spin-oscillator systems, measuring correlations of quadrature amplitudes or spin projections can reveal the same co-variation between variance and kurtosis predicted here. Cold-atom collective modes and optomechanical cavities already operate in regimes where cubic nonlinearities are measurable; in these systems, observing the correlated drift of low- and high-order moments would constitute a direct experimental validation of analytic rigidity across disparate physical domains.

\noindent\textit{Broader Implications -- }
This work recasts several pillars of continuous-variable quantum information. In resource-theoretic language, the Gaussian manifold defines the set of free operations, and rigidity forbids leaving this set without simultaneously altering the Gaussian statistics. Monotones based on skewness, kurtosis, or higher cumulants are therefore not independent; they evolve coherently along low-dimensional trajectories dictated by the underlying symplectic symmetry. For bosonic error-correcting codes this means that Hamiltonian-only encodings cannot reshape the distribution tails independently of the bulk. Cat, binomial, and GKP codes \cite{GKP2001,Michael2016,Ofek2016,CampagneIbarcq2020} all rely on higher-moment shaping that, by the theorem, requires measurement or ancilla assistance. In continuous-variable QKD \cite{GrosshansGrangier2002,WeedbrookRMP2012}, variance and kurtosis often appear as separable noise parameters, but rigidity removes that freedom: no Hamiltonian process can vary one without the other, simplifying security analyses and eliminating hypothetical attacks based on unattainable independent cumulant control.

At a deeper level, rigidity delineates the continuous-variable analogue of the Gottesman--Knill boundary. Gaussian (symplectic) operations play the role of Clifford gates: they evolve the first and second moments linearly, preserve the Gaussian manifold, and admit efficient classical simulation. In the discrete setting, the Gottesman--Knill theorem identifies the Clifford group as the maximal set of unitaries whose action on Pauli operators remains linear and therefore simulable by polynomial resources \cite{GottesmanThesis,AaronsonGottesman2004}. Our analytic results reveal the precise continuous-variable counterpart of that boundary. Quadratic Hamiltonians produce generators that terminate at second order in the Moyal expansion, keeping the dynamics within the Gaussian manifold and permitting covariance-matrix simulation. Any smooth non-quadratic Hamiltonian introduces third- and higher-order derivatives, coupling the Gaussian and non-Gaussian sectors and capable of generating Wigner negativity---the continuous-variable analogue of non-Clifford ``magic.'' The appearance of these higher derivatives marks the analytic threshold at which efficient classical simulation breaks down and the Gottesman--Knill theorem ceases to apply. From this perspective, rigidity provides not only the dynamical origin of the Gottesman--Knill divide but also its analytic generalization: the transition from second- to third-order structure in the Moyal expansion defines the universal boundary between Gaussian control and non-Gaussian universality in continuous-variable quantum mechanics. In this sense, analytic rigidity supplies the missing differential statement that underlies both operational and computational boundaries in continuous variables.
It clarifies that the Gaussian–non-Gaussian divide is not merely a technological limitation but a theorem of analytic structure: 
there exists no smooth deformation of Hamiltonian control that crosses it continuously.

\noindent\textit{Outlook -- }
Analytic rigidity carries implications for quantum complexity. Crossing the second-to-third-order threshold in the Moyal expansion corresponds to the onset of computational “magic” in the continuous-variable setting—the point where classical simulation based on covariance matrices ceases to apply. The degree of derivative coupling can thus serve as a quantitative analogue of stabilizer-rank monotones used in discrete systems, linking the geometry of Hamiltonian flows to the resources required for quantum advantage.

Our results also naturally raise several  several questions. Can similar rigidity principles appear in discrete or hybrid systems governed by smooth group actions on probability simplices? Is there a resource-theoretic metric that measures the curvature of the Gaussian manifold directly from experimental data? And can engineered measurement-based feedback circumvent analytic rigidity by introducing non-smooth dynamics? Pursuing these questions will clarify how the analytic structure of quantum mechanics delineates not only what can be simulated but also what can be physically realized.

This work was supported by Los Alamos National
Laboratory LDRD program grant 20230865PRD3.

\end{document}